\documentclass[runningheads]{llncs}
\usepackage[utf8]{inputenc}
\usepackage{amsmath, amssymb, thm-restate}
\usepackage{graphicx}
\usepackage[x11names]{xcolor}
\usepackage{todonotes}
\usepackage{complexity}
\usepackage{hyperref}
\usepackage{cite}
\usepackage[export]{adjustbox}
\usepackage{tabularx,multirow,makecell,cellspace,colortbl}
\usepackage{pdflscape}
\usepackage{cleveref}
\usepackage{wrapfig}
\usepackage{xpatch}
\usepackage{apptools}
\usepackage{tikz}
\usepackage{subcaption}
\usetikzlibrary{calc}

\usepackage{url,hyperref}
\hypersetup{
	hidelinks,
	colorlinks=true,
	citecolor=[rgb]{0.121 0.47 0.705},
	linkcolor=[rgb]{0.121 0.47 0.705},
	urlcolor=[rgb]{0.121 0.47 0.705}
}

\Crefname{figure}{Fig.}{Figs.}
\Crefname{observation}{Observation}{Observations}

\newsavebox{\smallereximagebox}
\newsavebox{\straightlineimagebox}

\makeatletter
\g@addto@macro\bfseries{\boldmath}
\makeatother

\let\oldendproof\endproof
\renewcommand\endproof{~\hfill$\qed$\oldendproof}
\spnewtheorem{observation}{Observation}{\bfseries}{\itshape}

\graphicspath{{figures/}}

\newcolumntype{Y}{S{>{\centering\arraybackslash}X}}
\newcolumntype{R}{>{\columncolor{LightBlue1}}Y}
\newcolumntype{L}{>{\columncolor{LightGoldenrod1}}Y}

\newboolean{arranging}
\setboolean{arranging}{false}
\newcommand{\aclearpage}{\ifthenelse{\boolean{arranging}}{\clearpage}{}}

\newboolean{showappx}
\newboolean{showappxnec}
\setboolean{showappx}{true}
\setboolean{showappxnec}{true}
\newcommand{\tablevalue}[4][]{\ifthenelse{\boolean{showappx} \OR \(\boolean{showappxnec} \AND \equal{#1}{N}\)}{\ensuremath{#2#3#4}}{\ensuremath{#2}
	}}

\newbox{\myorcidauthbox}
\sbox{\myorcidauthbox}{\large\includegraphics[height=1.7ex]{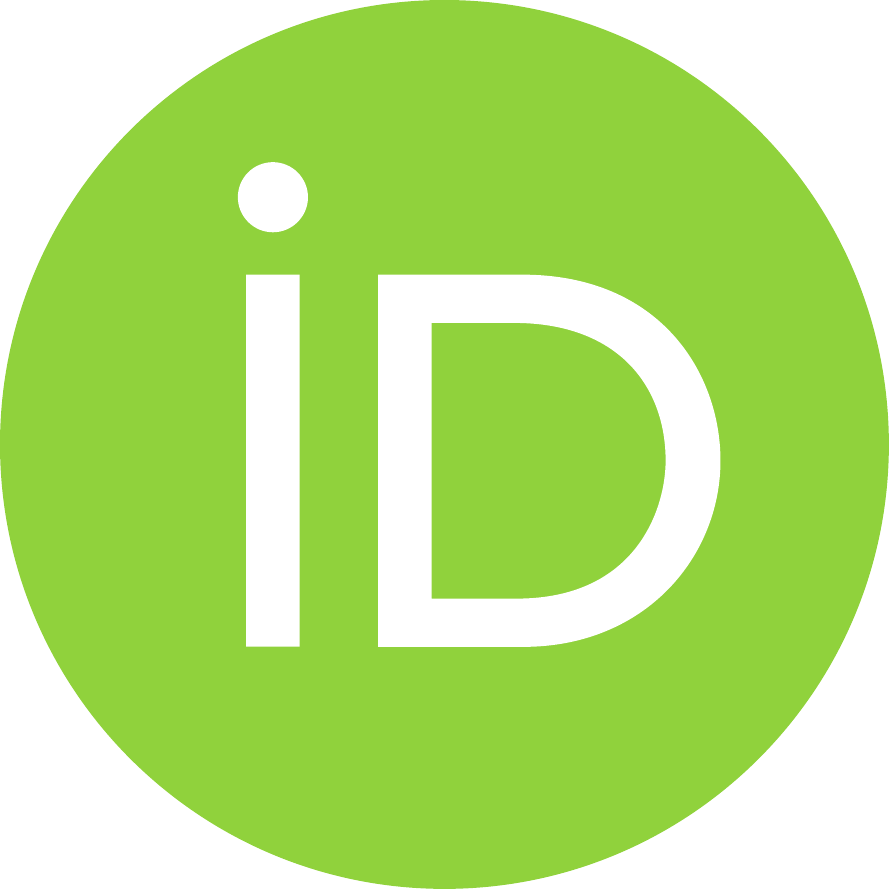}}
\newcommand{\orcid}[1]{\href{https://orcid.org/#1}{\usebox{\myorcidauthbox}}}

\author{Fabian Klute\inst{1}\orcid{0000-0002-7791-3604} \and
	Irene Parada\inst{2}\orcid{0000-0003-3147-0083}}
\authorrunning{F. Klute \and I. Parada}
\institute{Utrecht University, The Netherlands \\
	\email{f.m.klute@uu.nl}\and
	TU Eindhoven, The Netherlands\\
	\email{i.m.de.parada.munoz@tue.nl}}

\title{Saturated $ k $-Plane Drawings with Few Edges}

\begin{document}
	\maketitle
	
	\begin{abstract}
		A drawing of a graph is $k$-plane if no edge is crossed more than $k$ times.
		In this paper we study saturated $k$-plane drawings with few edges.
		These are $k$-plane drawings in which no edge can be added without violating $k$-planarity.
		For every number of vertices $n > k + 1$, 
		we present a tight construction with $\frac{n-1}{k+1}$ edges 
		for the case in which the edges can self-intersect. 
		If we restrict the drawings to be $\ell$-simple we show that 
		the number of edges in saturated $k$-plane drawings must be larger.
		We present constructions with few edges for different values of $k$ and $\ell$.
		Finally, we investigate saturated straight-line $k$-plane drawings.

		\keywords{saturated drawings  \and $k$-planarity \and simple drawings}
	\end{abstract}
	
	\section{Introduction}
	Given a simple undirected graph $G$ on $n$ vertices, 
	a \emph{drawing} of $G$ in the plane maps 
	each vertex of $G$ to a distinct point and
	each edge to a curve connecting the two corresponding points.
	In the literature, it is usually required that the edges are drawn as Jordan arcs. 
	However, in order to understand the influence of self-intersections in our study, 
	we also consider drawings in which the edges 
	are allowed to self-intersect. 
	In all our drawings we forbid edges to pass through vertices,
tangencies, and 
	three edges to properly cross in the same point.

	A drawing is \emph{$k$-plane} if no edge is properly crossed more than $k$ times. These drawings have received a lot of attention in the graph drawing community~\cite{didimoSurveyGraph2019,kobourovAnnotatedBibliography2017}.  
In this paper we study \emph{saturated} $k$-plane drawings, i.e.,
	$k$-plane drawings to which no edge can be added without violating the $k$-planarity of the drawing.

	A drawing is \emph{$\ell$-simple} if no edge crosses itself and two edges share at most $\ell$ points including their endpoints.
	Kyn\v{c}l et al.~\cite{kynclSaturatedSimple2015} initiated the study of saturated  $\ell$-simple drawings, i.e.
	$\ell$-simple drawings to which no edge can be added without violating the $\ell$-simplicity of the drawing.
They presented bounds on the minimum number of edges in saturated $\ell$-simple drawings. 
Recently the bounds for simple and 2-simple drawings were improved by Hajnal et al.~\cite{hajnalSaturatedSimple2018}.
Saturated drawings with few edges have also been studied by Aichholzer et al.~\cite{aichholzerEdgeVertexRatio2019} in the context of \emph{thrackles}, 
	that are drawings in which every pair of edges intersects exactly once, 
	either at a proper crossing or at a common endpoint. 
	
In this paper we study the minimum number of edges in saturated $k$-plane drawings 
	and put it in relation with different drawing restrictions, including $\ell$-sim\-plic\-i\-ty. 
	Previously, saturated simple $k$-plane drawings have been studied for $k=1$ and $2$.
	Brandenburg et al.~\cite{brandenburgDensityMaximal2013} showed 
	that every saturated simple $1$-plane drawing has at least $2.1n$ edges and 
	that there are arbitrarily large saturated simple $1$-plane drawings with $\frac{7n}{3}$ edges.
	For this problem, Bar{\'a}t and T{\'o}th~\cite{baratImprovementsDensity2018} recently improved the lower bound to $\frac{20n}{9}$.
	Auer et al.~\cite{auerSparseMaximal2013} gave a saturated simple $2$-plane drawing
	with at most $\frac{4n}{3}$ edges.
	Extending this research to larger values of $k$ was suggested by T{\'o}th as an interesting question~\cite[Section~$3.2$~Problem~$15$]{hongBeyondPlanarGraphs2017}.

\begin{table}[t]
			\caption{New and known bounds for the infimum ratio of edges to vertices for arbitrarily large saturated $k$-plane drawings (disregarding lower order terms). 
 			Numbers in the table are upper bounds while the colors represent the lower bounds. For straight-line drawings no lower bounds are claimed. Arrows indicate where the best known bound for higher simplicity coincides with the one for lower simplicity.
 			Citations and references to theorems point to where the result is proven.
	}
	\newcommand{\tabequal}{\multirow{-2}{*}{$\rightarrow$}}		
	\begin{tabularx}{\textwidth}{ X|R|*{5}{L}|L }			
	\hline
	\rowcolor{white}
	\makecell[l]{Upper\\ bounds} & 
	\makecell{with self-\\crossings} & 
	\makecell{$k+1$-\\simple} & 
	\makecell{$4$-simple} & 
	\makecell{$3$-simple} & 
	\makecell{$2$-simple} & 
	\makecell{simple} & 
	\makecell{straight-\\line} \\
	\hline
	
	\multirow{2}{*}{$1$-plane} & 
	$\tablevalue{\frac{1}{2}}{=}{0.5}$ & 
	&
	&
	&
	\makecell{$\tablevalue{\frac{3}{2}}{=}{1.5}$} & 
	\cellcolor{Goldenrod1}\makecell{$\tablevalue[N]{\frac{7}{3}}{\approx}{2.33}$} & 
	\cellcolor{white}\makecell{$\tablevalue{\frac{11}{5}}{=}{2.2}$} \\
	
	&
	\makecell{Thm.\ref{thm:selfcrossings}}&
	\tabequal & 
	\tabequal & 
	\tabequal & 
	\makecell{Thm.\ref{thm:lowsimplicity}} &
	\cellcolor{Goldenrod1}\makecell{\cite{brandenburgDensityMaximal2013};Thm.\ref{thm:kpsimple}} &
	\cellcolor{white}\makecell{Thm.\ref{thm:kpsl}}
	\\
	\hline
	
	\multirow{2}{*}{$2$-plane} & 
	$\tablevalue{\frac{1}{3}}{\approx}{0.33}$ & 
	&
	&
	$\tablevalue{\frac{2}{3}}{\approx}{0.66}$ & 
	$\tablevalue{\frac{4}{5}}{=}{0.8}$ & 
	\makecell{$\tablevalue{\frac{4}{3}}{\approx}{1.33}$} & 
	\cellcolor{white}$\tablevalue{\frac{3}{2}}{=}{1.5}$ \\
	
	&
	\makecell{Thm.\ref{thm:selfcrossings}}&
	\tabequal & 
	\tabequal & 
	\makecell{Thm.\ref{thm:noselfcrossings}} &
	\makecell{Thm.\ref{thm:lowsimplicity}} &
	\makecell{\cite{auerSparseMaximal2013};Thm.\ref{thm:kpsimple}} &
	\cellcolor{white}\makecell{Thm.\ref{thm:kpsl}}
	\\
	\hline
	
	\multirow{2}{*}{$3$-plane} & 
	$\tablevalue{\frac{1}{4}}{=}{0.25}$ & 
	&
	$\tablevalue{\frac{2}{3}}{\approx}{0.66}$ & 
	\makecell{$\tablevalue{\frac{3}{4}}{=}{0.75}$} & 
	\makecell{$\tablevalue{\frac{4}{5}}{=}{0.8}$} & 
	$1$ & 
	\cellcolor{white}$\tablevalue{\frac{7}{6}}{=}{1.2}$ \\
	
	&
	\makecell{Thm.\ref{thm:selfcrossings}}&
	\tabequal & 
	\makecell{Thm.\ref{thm:noselfcrossings}}&
	\makecell{Thm.\ref{thm:smaller}}&
	\makecell{Thm.\ref{thm:lowsimplicity}} &
	\makecell{Thm.\ref{thm:kpsimple}} &
	\cellcolor{white}\makecell{Thm.\ref{thm:kpsl}}
	\\
	\hline
	
	\multirow{2}{*}{$4$-plane} & 
	$\tablevalue{\frac{1}{5}}{=}{0.2}$ & 
	&
	\makecell{$\tablevalue{\frac{2}{5}}{=}{0.4}$} & 
	\makecell{$\tablevalue[N]{\frac{3}{7}}{\approx}{0.43}$} & 
	$\tablevalue[N]{\frac{6}{13}}{\approx}{0.46}$ & 
	$\tablevalue{\frac{2}{3}}{\approx}{0.66}$ & 
	\cellcolor{white}$ 1 $ \\
	
	&
	\makecell{Thm.\ref{thm:selfcrossings}}&
	\tabequal & 
	\makecell{Thm.\ref{thm:smaller}}&
	\makecell{Thm.\ref{thm:smaller}}&
	\makecell{Thm.\ref{thm:lowsimplicity}} &
	\makecell{Thm.\ref{thm:kpsimple}} &
	\cellcolor{white}\makecell{Thm.\ref{thm:kpsl}}
	\\
	\hline
	
	\multirow{2}{*}{$k$-plane} & 
	\makecell{$\frac{1}{k+1}$} & 
	\makecell{$\frac{1}{\left\lfloor\frac{k}{2}\right\rfloor + \frac{1}{2}}$} & 
	& 
	& 
	\makecell{$\frac{2}{k + \frac{2}{k+2}}$} & 
	\makecell{$\frac{2}{k-1}$} & 
	\cellcolor{white}\makecell{$\frac{4k + 2}{k^2+2}$} \\
	
	&
	\makecell{Thm.\ref{thm:selfcrossings}}&
	\makecell{Thm.\ref{thm:noselfcrossings}} &
	\tabequal &
	\tabequal &
	\makecell{Thm.\ref{thm:lowsimplicity}} &
	\makecell{Thm.\ref{thm:kpsimple}} &
	\cellcolor{white}\makecell{Thm.\ref{thm:kpsl}}
	\\
	\hline
	\end{tabularx}
	\begin{tikzpicture}
\coordinate (T) at (0,1);
		\coordinate (A) at (1.5,1);
		\coordinate (B) at ($(A) + (2.5,0)$);
		\coordinate (C) at ($(B) + (2.5,0)$);
		\node at (T) {\normalsize Lower bounds:};
		\node (rect) at (A) [fill,LightBlue1,minimum width=.5cm,minimum height=.3cm] {};
		\node (rect) at (B) [fill,LightGoldenrod1,minimum width=.5cm,minimum height=.3cm] {};
		\node (rect) at (C) [fill,Goldenrod1,minimum width=.5cm,minimum height=.3cm] {};	
		
		\node at ($(A) + (1.05,0)$) {$\frac{1}{k+1}$~Thm.\ref{thm:selfcrossings}};
		\node at ($(B) + (1.05,0)$) {$\frac{2}{k+2}$~Thm.\ref{thm:noselfcrossings}};
		\node at ($(C) + (1.15,0)$) {$\frac{20}{9} \approx 2.22$~\cite{baratImprovementsDensity2018}};	
	\end{tikzpicture}
		\label{tbl:overview}
	\end{table}

	\paragraph{Relation to independent similar work.}
	After submission of the first version of our work, we learned that, independently,
	Chaplick et al.~\cite{chaplickEdgeMinimumSaturated20} had investigated saturated drawings of sparse $k$-planar (multi-)graphs.
The main difference is that their setting allows for multiple parallel edges and
	therefore no two vertices can lie in the same cell regardless of whether they are connected in the graph.
	This leads to different results, even though the constructions and proof techniques are similar.

	One of the closest results is for saturated $k$-plane drawings without further restrictions. 
	Using a similar construction to ours in Section~\ref{sec:norestrictions},
	they also obtain a tight bound in their non-restricted setting.
	In their conclusion they describe and illustrate a construction obtaining an $\frac{1}{k+1}(n-1)$ upper bound for $k$-plane drawings of simple graphs without restricting the simplicity. 
	Their construction only applies for even values of $k$ and matches the bound 
	that we obtain in Section~\ref{sec:noselfcrossings}.

	After reading their paper, we included a new section, Section~\ref{sec:matchings},
	resolving an open question asked by Chaplick et al. on
	the existance of saturated simple $k$-plane drawings of matchings for $k = 4$,~$5$, and~$6$
	when inserting parallel edges is allowed.
	We show that in this setting no saturated simple $4$- and $5$-plane drawings of matchings exist. 
	In contrast, we provide a construction for saturated simple $6$-plane drawings of arbitrarily large matchings
	in which no (parallel) edge can be inserted.

\paragraph{Our results.}
We study the infimum ratio of edges to vertices for arbitrarily large saturated $k$-plane drawings. 
	For different restrictions on the drawing, we present both lower bounds and constructions with few edges that give upper bounds. 
	\Cref{tbl:overview} summarizes our results. 
In \Cref{sec:norestrictions} we show that if we allow self-intersecting edges 
there are saturated $k$-plane drawings with $n$ vertices and $\frac{n-1}{k+1}$ edges. 
	Moreover, we prove that the edge-vertex ratio in this construction is the smallest possible.
	We begin our paper by encapsulating a helpful operation for our constructions, called \emph{stashing},
	in Section~\ref{sec:stashing} 
	In \Cref{sec:noselfcrossings} we show that if we disallow self-intersections we need at least two times the amount of edges: 
	A saturated $k$-plane drawing with $n$ vertices has at least $\frac{2(n-1)}{k+1}$ edges. 
	In this setting, we present a construction with $n$ vertices and $\frac{2(n-1)}{k+2}$ edges. 
	However, this construction requires a $k+1$-simple drawing.
	In \Cref{sec:restricting} we focus on drawings with low simplicity. 
	For 2-simple drawings we give a construction with $n$ vertices and less than $\frac{2(n-1)}{k}$ edges. 
	Then, adapting a construction by Auer et al.~\cite{auerSparseMaximal2013}, 
	we obtain simple drawings with $n$ vertices and $\frac{2n}{k-1} + O(k)$ edges.  
	In \Cref{sec:straightline} we consider straight-line drawings 
	and present a construction with $n$ vertices and $\frac{4k + 2}{k^2+2}(n-1)$ edges.
	Finally, in \Cref{sec:matchings} we study saturated simple $k$-plane drawings of matchings
	in the setting in which inserting parallel edges must also be prevented.
	Our constructions close the gap left open by Chaplick et al.~\cite{chaplickEdgeMinimumSaturated20}.

\section{Stashing} \label{sec:stashing}
	Several proofs throughout the paper use stashing into free cells to construct upper bounds 
	to the minimum edge-vertex ratio in a saturated $k$-plane drawing. 
	The next lemma provides a tool to more easily compute this bound when we stash isolated 
	vertices into free cells; compare \Cref{thm:selfcrossings}.
	\begin{restatable}{lemma}{lemstashing}
		\label{lem:stashing}
		For $k>1$, let $D_0$ be a saturated $k$-plane drawing with $n_0$ vertices, $m_0$ edges, and
		$f_0 > 0$ free, non-intersecting cells, then there exist arbitrarily large saturated $k$-plane drawings on $n$ vertices with
		$m = \frac{m_0}{n_0+f_0- 1}(n-1)$ edges.
\end{restatable}
	\begin{proof}
		Let $D_0$ be the initial saturated $k$-plane drawing.
		The drawing $D_1$ is defined by stashing $k-1$ isolated vertices into $k-1$ free cells of $D_0$. 
		Note that $D_1$ has one free cell. 
		Let $t > 1$ be an integer and $D_t$ the drawing obtained 
		by stashing $D_1$ into the only free cell of $D_{t-1}$. 
		For all $t\ge 1$ we define $D'_t$ to be the drawing resulting from stashing one vertex into the only free cell of $D_t$.
		Let $n$ be the number of vertices in $D'_t$ with $t\ge 1$ 
		and $m$ the number of edges.
		We obtain that $m = tm_0$ and $n = n_0t + (f_0-1)t + 1 = (n_0 + f_0 - 1)t + 1$. 
		Rearranging gives us that $t = \frac{n - 1}{n_0 + f_0 - 1}$ and consequently
		$m = \frac{m_0(n-1)}{n_0 + f_0 - 1}$.
\end{proof}
	
	It can also be helpful to not stash isolated vertices,
	but instead replace them with two new vertices and a new edge completely drawn inside a free cell.
	The following lemma establishes the resulting upper bound and shows that using new edges instead of isolated vertices
	leads to a lower edge-vertex ratio if and only if $m_0 > n_0 + f_0 - 1$.
	In other words if and only if the edge-vertex ratio resulting from stashing isolated vertices exceeds one.
	
	\begin{lemma}
		\label{lem:stashingedge}
		For $k>1$, let $D_0$ be a saturated $k$-plane drawing with $n_0$ vertices, $m_0$ edges, and $f_0 > 0$ free cells,
		there are arbitrarily large saturated $k$-plane drawings on $n$ vertices with $m = \frac{m_0 + f_0 - 1}{n_0 + 2f_0 - 2}(n-1)$ edges
		and if $m_0 > n_0 + f_0 - 1$ this drawing has less edges than a saturated $k$-plane drawing obtained from $D_0$
		with Lemma~\ref{lem:stashing}.
	\end{lemma}
	\begin{proof}
		Let $D_0$ be the initial saturated $k$-plane drawing.
		The drawing $D_1$ is defined by stashing $k-1$ new edges connecting two new vertices into $k-1$ free cells of $D_0$. 
		Note that $D_1$ has one free cell. 
		Let $t > 1$ be an integer and $D_t$ the drawing obtained 
		by stashing $D_1$ into the only free cell of $D_{t-1}$. 
		For all $t\ge 1$ we define $D'_t$ to be the drawing resulting from stashing one vertex into the only free cell of $D_t$.
		Let $n$ be the number of vertices in $D'_t$ with $t\ge 1$ 
		and $m$ the number of edges.
		We obtain that $m = t(m_0 + f_0 - 1)$ and $n = n_0t + 2(f_0-1)t + 1 = (n_0 + 2f_0 - 2)t + 1$. 
		Rearranging gives us that $t = \frac{n - 1}{n_0 + 2f_0 - 2}$ and consequently
		$m = \frac{m_0 + f_0 - 1(n-1)}{n_0 + 2f_0 - 2}$.
		
		For the second part of the lemma let $D_S$ be a saturated $k$-plane drawing on $n$ vertices with $m_S$ edges
		obtained from $D_0$ with Lemma~\ref{lem:stashing} and
		$D_I$ a saturated $k$-plane drawing obtained as above with $n$ vertices and $m_I$ edges.
		We are interested when $m_I > m_S$:
		\begin{eqnarray*}
			\frac{m_0}{n_0 + f_0 - 1}(n-1) &>& \frac{m_0 + f_0 - 1}{n_0 + 2f_0 - 2}(n-1) \\
			m_0(n_0 + 2f_0 - 2) &>& (n_0+f_0-1)(m_0 + f_0 -1) \\			 
			m_0f_0 - m_0 &>&  f_0n_0 - n_0 + f_0^2- 2f_0  + 1 \\
			m_0 &>&  n_0 + f_0 - 1.
		\end{eqnarray*}
	\end{proof}

	\section{Allowing Self-Intersections}
	\label{sec:norestrictions}
	In this section we consider $k$-plane drawings without any additional restrictions.
	In particular, we allow edges to self-intersect.
	As a consequence, the boundary of one cell may consist of only one crossing and one edge segment. 
	If the edge to which this edge segment belongs has $k$ crossings, 
	we can use such a cell to place a vertex in it. 
	In general, in a $k$-plane drawing we say that an edge is \emph{saturated} if it is crossed $k$ times. 
	A cell without vertices on its boundary and 
	bounded only by segments of saturated edges is called \emph{free}.
	In a saturated $k$-plane drawing a vertex in a free cell will be isolated. 
	In this way we can produce a $k$-plane drawing with low edge-vertex ratio: 
	making an edge self-intersect $k$ times and placing an isolated vertex in every free cell
	we obtain a saturated $k$-plane drawing with one single edge and $k+2$ vertices. 
	
	To be able to produce arbitrarily large drawings, 
	we can use one free cell 
	and, instead of placing a vertex, recursively place the construction. 
	In general, given a saturated $k$-plane drawing $D$ with a free cell $c$, 
	\emph{stashing} a drawing $D'$ into $c$ refers to producing a new drawing 
	that consists of the union of the two drawings, with $D'$ drawn inside $c$.

	\begin{figure}[t]
		\centering
		\includegraphics[page=1]{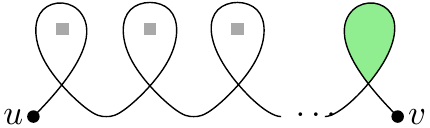}
		\caption{One edge with self-intersections. Gray squares represent isolated vertices that we stash and the green cell is used for recursively stashing the construction.}
		\label{fig:selfcrossing}
	\end{figure}
	
 	\begin{theorem}
		\label{thm:selfcrossings}
		There are arbitrarily large drawings on $n$ vertices with $\frac{n-1}{k+1}$ edges,
		which are saturated $k$-plane. Moreover, this bound is tight.
		
\end{theorem}
	\begin{proof}
		Let $D_0$ be the initial drawing consisting of one edge with $k$ self-in\-ter\-sec\-tions defining $k$ non-intersecting free cells. 
		The drawing $D_1$ is defined by stashing $k-1$ isolated vertices into $k-1$ free cells of $D_0$; see \Cref{fig:selfcrossing}. 
		Note that $D_1$ has one free cell. 
		Let $t > 1$ be an integer and $D_t$ the drawing obtained 
		by stashing $D_1$ into the only free cell of $D_{t-1}$. 
		For all $t\ge 1$ we define $D'_t$ to be the drawing resulting from stashing one vertex into the only free cell of $D_t$.
		Let $n$ be the number of vertices in $D'_t$ with $t\ge 1$ 
		and $m$ the number of edges.
		We obtain that $m = t$ and $n = 2t + (k-1)t + 1 = (k+1)t + 1$.  
Rearranging the equation we get that $t = \frac{n - 1}{k+1}$ and consequently $m = \frac{n - 1}{k+1}$.

		Let $D=D(G)$ be a saturated $k$-plane drawing of a graph $G$ on $n$ vertices 
		and let $x$ be the number of crossings in $D$.
		Since $G$ might not be connected, we denote with $\gamma$ the number of connected components of $G$.
		Consider the planarization $\mathcal D = (P,C)$ obtained from $D$ 
		by replacing every crossing in $D$ with a vertex and 
		every edge segment with an edge.
		Observe that $\mathcal D$ has
		at most one self-loop per vertex and 
		at most one pair of parallel edges between each two vertices.
		In the following let $\gamma' \leq \gamma$.
		
		To prove the desired lower bound we make use of Euler's formula, accounting also for the connected components: $|P| - |C| + f = \gamma' + 1,$
		where $f$ includes the number of faces defined by self-loops and multiple edges in $\mathcal D$.
		Since the formula is usually stated for connected, simple planar graphs we 
		include the easy details on how this version can be derived in \Cref{apx:euler}.

		Next, we count the number of vertices in $P$ and edges in $C$.
		Since we added one vertex to $P$ for every vertex and crossing in $D$
		we get that $|P| = n + x$.
		To count the edges in $C$, traverse every edge in $D$ from one of its endpoints to the other one.
		For each crossing we see along this traversal
		there is one edge in~$C$ plus one additional edge for the last edge segment.  
Since every crossing, also a self-intersection,
		was seen twice during this traversal we obtain that $|C| = m + 2x$. 
		As a result we get from Euler's formula that $n - m - x + f = \gamma' + 1$.

		Using that there are at most $k\cdot m$ many crossings we get that $f + n \leq \gamma' + 1 + (k+1) m$.
Observe, that two non-adjacent vertices cannot share a cell in~$D$ without
        contradicting the assumption that $D$ is saturated.
Hence, we find $f \geq \gamma'$:
		Each connected component of $\mathcal D$ has at least one vertex and 
		no two vertices of different components can lie on the same face.
		Consequently, $\gamma' + n \leq f + n \leq \gamma' + 1 + (k+1) m$,  
which yields $m \geq \frac{n-1}{k+1}$ as desired.		
 	\end{proof}

	\section{Disallowing Self-Intersections}
	\label{sec:noselfcrossings}
In \Cref{sec:norestrictions} we saw that allowing self-intersections leads to very few edges being
	necessary to create saturated $k$-plane drawings.
	In this and all following sections we consider only $k$-plane drawings without self-intersections. 
The price we pay is that the number of edges we need is roughly doubled.

	For the best construction that we have without self-intersections, 
	the drawing that we use for stashing has two edges forming a spiral; see \Cref{fig:noself}. 
	We then can use one free cell to recursively stash the drawing 
	and the rest to stash vertices. 
	Intuitively, in our drawing we form almost one free cell for every crossing and
	every edge is crossed $k$ times.
	However, an essential difference with the previous section is that here 
each crossing counts for two edges.

	\begin{figure}[b]
		\centering
		\begin{minipage}[t]{.45\textwidth}
			\centering
			\includegraphics[page=1]{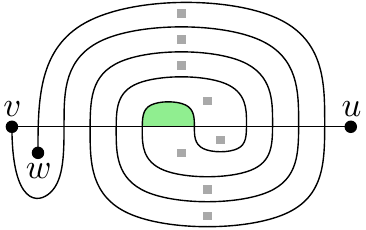}
			\subcaption{Construction for $k$ even}
			\label{fig:noselfeven}
		\end{minipage}
\begin{minipage}[t]{.45\textwidth}
			\centering
			\includegraphics[page=2]{no_self_crossings}
			\subcaption{Construction for $k$ odd}
			\label{fig:noselfodd}
		\end{minipage}
		\caption{Construction used in \Cref{thm:noselfcrossings}. 
			Gray squares represent isolated vertices that we stash and green cells are used for recursively stashing the construction.}
		\label{fig:noself}
	\end{figure}

	We split the proof into three lemmas.
	\Cref{lem:upkpnoselfcreven,lem:upkpnoselfcrodd} establish the upper bound on the edge-vertex ratio
	for the cases of $k$ even and odd, respectively.
	In \Cref{lem:lowkpnoselfcro} we prove the lower bound.
	\begin{lemma}
		\label{lem:upkpnoselfcreven}
		For even $k > 1$ there are arbitrarily large drawings on $n$ vertices
		with $m = \frac{2(n-1)}{k+1}$ edges, which are
		saturated, $k$-plane, and in which no edge self-intersect.
	\end{lemma}
	\begin{proof}
		We construct a drawing $D_0$ as follows. 
		Take a path on three vertices and let $e = uv, f = vw$ be the two edges and $u,v,w$ the three vertices.
		Draw $e$ as a straight line in the plane.
		Then we draw $f$ as a spiral, intersecting $e$ exactly $k$ times as shown in \Cref{fig:noselfeven}.
		Crucially, since $k$ is even we can cross such that  
		$u$ and $v$ are incident to the outer cell, 
		while $w$ is placed in the other cell which is incident to $v$.
		As $e$ and $f$ share $k$ crossings we cannot add the edge between $u$ and $w$ 
		without crossing $e$ or $f$ the $k+1$-st time.
		Hence, the drawing is saturated, $k$-plane, and has no edge self-intersect.
		Furthermore, there are $k-1$ free cells.
		Using \Cref{lem:stashing} we get that for any number of vertices $n \geq 3$
		the number of edges $m$ is $\frac{2(n-1)}{3 + k - 1 - 1} = \frac{2(n-1)}{k + 1}$.
	\end{proof}
	
	\begin{restatable}{lemma}{lemupkpnoselfcrodd}
		\label{lem:upkpnoselfcrodd}
		For odd $k > 1$ there are arbitrarily large drawings on $n$ vertices
		with $m = \frac{2(n-1)}{k}$ edges, which are
		saturated, $k$-plane, and in which no edge self-intersect.
	\end{restatable}
	\begin{proof}
		We construct a drawing $D_0$ nearly as in the proof of \Cref{lem:upkpnoselfcreven}.
		An illustration is shown in \Cref{fig:noselfodd}.
		The only difference is that, since $k$ is odd, the $k-1$-th crossing of $e$ and $f$
		is such that $f$ is on the lower side of the supporting line through $e$,
		while the last crossing of $f$ with $e$ is from the upper side.
		Consequently, $f$ has to be drawn such that it creates a cell in which we either enclose $u$ or $w$.
		As a result there are only $k-2$ free cells and 
		for $n$ vertices we obtain $\frac{2(n-1)}{3 + k - 2 - 1}(n - 1) = \frac{2(n-1)}{k}$ as the number of edges.
	\end{proof}
	
	\begin{restatable}{lemma}{lemlowkpnoselfcro}
		\label{lem:lowkpnoselfcro}
		Any saturated $k$-plane drawing on $n$ vertices in which no edge self-intersects 
		has at least $\left\lfloor\frac{2n-1}{k+2}\right\rfloor$ edges.	
	\end{restatable}
	\begin{proof}
		Given some saturated $k$-plane drawing $D(G)$ of a graph $G$ with $n$ vertices, $m$ edges, and $x$ crossings
		in which no edge self-intersects,
		let $\mathcal D = (P,C)$ be the planarization of $D$.
		Furthermore, let $\gamma$ be the number of connected components in $G$ and 
		$\gamma'$ the number of connected components in $\mathcal D$.
		Observe that $\gamma' \leq \gamma$.
		
		In the following we proceed as in the the proof of the lower bound in \Cref{thm:selfcrossings}. 
		Note that since in $D$ no edge self-intersects, 
		we find that $\mathcal D$ has not self-loops. 
		Nonetheless, arguing as before we again obtain that
		$|P| = n +x$ and $|C| = m + 2x$.
		Using Euler's formula as derived in \Cref{apx:euler}
		we obtain that 
		\begin{eqnarray*}
			n - m - x + f &=& \gamma' + 1\\
			f + n &=& \gamma' + 1 + m + x.
		\end{eqnarray*} 
		Since no edge has self-intersections in $D$ the number of crossings is upper bounded by $\frac{km}{2}$.
		Consequently we get that 
		\begin{eqnarray*}
			f + n &\leq& \gamma' + 1 + \frac{k + 2}{2}m.
		\end{eqnarray*} 
		Again we can argue that no two non-adjacent vertices share a cell and hence $f \geq \gamma'$ holds.
		Finally, we obtain 
		\begin{eqnarray*}
			\gamma' + n \leq f + n &\leq& \gamma' + 1 + \frac{k + 2}{2}m
		\end{eqnarray*}
		which yields $m \geq \frac{2n-1}{k + 2}$.
	\end{proof}
	
	\Cref{lem:upkpnoselfcreven,lem:upkpnoselfcrodd,lem:lowkpnoselfcro} proof \Cref{thm:noselfcrossings}.
	
	\begin{restatable}{theorem}{thmnoselfcrossingsub}
		\label{thm:noselfcrossings}
		For every $k > 1$ there are arbitrarily large drawings on $n$ vertices with 
		$\frac{n-1}{\left\lfloor\frac{k}{2}\right\rfloor + \frac{1}{2}}$ edges,
		which are saturated $k$-plane and in which no edge self-intersects.
		Moreover, any saturated $k$-plane drawing in which no edge self-in\-ter\-sec\-ts
		has at least $\left\lfloor\frac{2n-1}{k+2}\right\rfloor$ edges.
\end{restatable}

	\aclearpage
	
	\section{Restricting the Simplicity}
	\label{sec:restricting}
	
	The construction in \Cref{sec:noselfcrossings}  
	of families of saturated $k$-plane drawings 
	requires two adjacent edges to cross $k$ times. 
	Thus, the resulting drawings are $k+1$-simple but not $k$-simple. 
	In fact, any $k$-plane drawing is $k+1$-simple. 
	However, the best studied $k$-plane drawings are the simple ones. 
	By slightly increasing the edge-vertex ratio, 
	in this section 
	we present families of saturated $k$-plane drawings 
	with few edges
	that are $2$-simple and simple.

	\subsection{$2$-Simple Drawings}
	\label{sec:2simple}

Our construction of families of saturated $2$-simple $1$-plane drawings 
is illustrated in \Cref{fig:1pnoself}. 
As in the previous sections, it is based on recursive stashing. 
However, in this case it is done by repeatedly making a copy of the drawing in \Cref{fig:1pnoself} and identifying the bottommost edge with the green edge of the previous copy. 

For $2$- and $3$-planarity our $2$-simple construction is illustrated in \Cref{fig:circlelowsimple}. 
For the $2$-plane construction we use the orange solid sub-arcs 
while for $3$-plane construction we use the orange dotted sub-arcs. 
The final drawings are obtained by stashing as in the previous sections. 
To construct arbitrarily large saturated $2$-simple $k$-plane drawings 
we start from the drawing in \Cref{fig:circlelowsimple} 
with the orange solid sub-arcs.
We then insert two sets of $\frac{k-2}{2}$ independent edges 
crossing the two free cells of the drawing as in \Cref{fig:gadgetlowsimple}.

	We first prove the case of saturated $2$-simple $1$-plane drawings in \Cref{lem:2s1p} and
	then for saturated $2$-simple $2$- and $3$-plane drawings in \Cref{lem:2s2p3p}.
	Then we describe our construction for saturated $2$-simple $k$-plane drawings with $k >3$ 
	and prove in \Cref{lem:rz} that the drawing has these properties as well as determine the number
	of vertices, edges, and free cells.
	
	\begin{restatable}{lemma}{lemtwosonep}
		\label{lem:2s1p}
		There are arbitrarily large drawings on $n$ vertices with $\frac{3(n-2)}{2}$ edges, 
		which are saturated, $2$-simple, and $1$-plane
	\end{restatable}
	\begin{proof}
		Consider the drawing in \Cref{fig:1pnoself}.
		The depicted graph consists of two triangles that share a vertex.
		Let $u,v,w,x,y$ be the five vertices and $u,v,w$ and $w,x,y$ the two triangles.
		The depicted drawing is obtained by introducing a crossing between 
		$uv$ and $vw$ and $wx$ and $wy$.
		The edges $uv$ and $xy$ are drawn plane and
		no further crossings are allowed in this drawing.
		Let $p$ be the crossing point between $uw$ and $vw$ and 
		$q$ the crossing point between $wx$ and $wy$.
		We draw the edge $xy$ into the cell completely bounded by the two edge segments between $q$ and $w$.
		The edge $uv$ is drawn into the outer cell.
		Let $D_0$ be this drawing.		
		Clearly, $D_0$ is $2$-simple and $1$-plane as there are only two independent crossings.
		Furthermore, every pair of non-adjacent vertices is separated from each other by the 
		edge segments between $q$ and $w$.
		Hence, $D_0$ is also saturated.
		
		Next, observe that there is one cell completely bounded by the edge segments between $p$ and $w$
		that is empty and has only one vertex, namely $w$, incident to its boundary.
		We draw a vertex into that cell and connect it with an edge to $w$ that is drawn inside of the cell;
		see the gray vertex in \Cref{fig:1pnoself}.
		Let the resulting drawing be $D_1$.
		
		Similarly to when we stash an isolated vertex into a free cell,
		we now recursively add copies of $D_1$.
		To obtain $D_2$ consider two copies of $D_1$, $D_1$ and $D'_1$.
		We identify the edge $xy$ in $D_1$ with the edge $uv$ in $D'_1$ and
		draw the remainder of $D'_1$ completely inside the cell bounded by the edge segments
		between $u$, $v$, and $w$.
		Finally, rename the vertices such that the vertices $u$, $v$, $w$, $x$, and $y$ in $D_2$ 
		correspond to the ones of $D'_1$ with the same name.
		Let $t > 1$ be an integer and $D_t$ the drawing obtained by 
		identifying the edge $uv$ in a copy of $D_1$ with the edge $xy$ in $D_{t-1}$ and
		drawing the remainder of $D_1$ completely inside the cell bounded by the edge segments
		between $x$, $y$, and $w$ in $D_{t-1}$.
		
		It remains to compute the number of edges in $D_t$.
		For each added copy of $D_1$ we find four vertices in $D_t$.
		Adding the two additional vertices of the last edge $xy$ in $D_t$
		which we do not identify with a copy of $D_1$ we get 
		that $D_t$ has $n = 4t + 2$ vertices in total.
		Similarly, $D_t$ has six edges for every added copy of $D_1$, 
		hence $m = 6t$ edges in total. 
		Hence, with $t = \frac{n - 2}{4}$ we get that 
		$m = 6\frac{(n - 2)}{4} = \frac{3(n - 2)}{2}$.
	\end{proof}
	
	\begin{figure}[tb]
		\centering
		\begin{minipage}[t]{.31\textwidth}
			\centering
			\includegraphics[page=1]{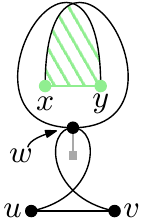}
			\subcaption{Construction for $k=1$}
			\label{fig:1pnoself}
		\end{minipage}		
		\begin{minipage}[t]{.3\textwidth}
			\centering
			\includegraphics[page=2]{2_simple}
			\subcaption{Construction for $k=2$ and $k = 3$ (dotted)}
			\label{fig:circlelowsimple}
		\end{minipage}
		\begin{minipage}[t]{.34\textwidth}
			\centering
			\includegraphics[page=1]{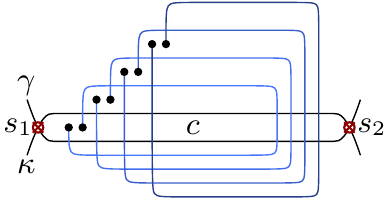}
			\subcaption{Edges added to the free cells of (b) for $k > 3$}
			\label{fig:gadgetlowsimple}
		\end{minipage}
		\caption{The drawings used in \Cref{thm:lowsimplicity}.
			Gray squares represent vertices that we stash and green cells are used for recursively stashing the construction.}
		\label{fig:lowsimple}
	\end{figure}	
	
	\begin{lemma}
		\label{lem:2s2p3p}
		For $k = 2,3$ there are arbitrarily large drawings on $n$ vertices with $\frac{4(n-1)}{5}$ edges, which are
		saturated, $2$-simple, and $k$-plane .
	\end{lemma}
	\begin{proof}
		First we show the lemma for $k = 2$
		Consider the drawing shown in \Cref{fig:circlelowsimple} ignoring the dotted variation.
		It consists of a cycle on four vertices.
		Let $u,v,w,x$ be those vertices and $uv$, $vw$, $wx$, and $ux$ the edges.
		We draw the edges such that 
		$uv$ and $wx$ and $vw$ and $ux$ each cross twice.
		Let $D_0$ be the resulting drawing 
		Clearly $D_0$ is $2$-plane and $2$-simple.
		Consequently non of the four edges can be crossed.
		Observe that the non-adjacent pairs of vertices, namely $u$ and $w$ and $v$ and $x$, 
		are incident to different cells.
		Hence, the drawing is saturated.
		Finally, we observe that there are two free cells
		formed by the crossings between $uv$ and $wx$ and $vw$ and $ux$.
		Consequently, the number of vertices in $D_0$ is $n_0 = 4$, the number of edges $m_0 = 4$, 
		and the number of free cells $f_0 = 2$
		Using \Cref{lem:stashing} we obtain that there are arbitrarily large
		saturated $2$-simple $2$-plane drawings on $n$ vertices with 
		$m = \frac{4n}{4 + 2 - 1} = \frac{4n}{5}$ edges.
		
		For the case of $k=3$ we modify the previous drawing $D_0$.
		The modification is illustrated in \Cref{fig:circlelowsimple} with the dotted variation.
		It consists of adding a crossing between the edges $uv$ and $vw$ and $wx$ and $ux$
		close to $v$ and $x$, respectively.
		Let $D'_0$ be the resulting drawing.
		Clearly the drawing is $3$-plane.
		It is also $2$-simple as no edge shares more than $2$ points with any other edge.
		Furthermore, the non-adjacent vertices $u$ and $w$ and $v$ and $x$ are still incident to different cells.
		Hence, $D'_0$ is a saturated $2$-simple $3$-plane drawing.
		Again, observe that therer are two free cells formed by $uv$ and $wx$ and $vw$ and $ux$ which yields that 
		$D'_0$ has the same number of vertices, edges, and free cells as $D_0$.
		Consequently, we obtain the same bound on the number of edges in arbitrarily large 
		saturated $2$-simple $3$-plane drawings
		as for the case of saturated $2$-simple $2$-plane drawings.
	\end{proof}
	
	To proof our bounds for $2$-simple $k$-plane drawings we require an additional construction.
It is a drawing consisting of $z$ independent edges that we draw in a certain way,
	such that it can be inserted into the drawing used in \Cref{lem:2s2p3p}.
	The construction will result in a drawing with $2z + 2$ crossings,
	hence the final drawing always has even planarity.
	
	Let $G_z$ be a graph consisting of $z>0$ independent edges.
	In the following we describe how to construct the drawing $M_z$ from $G_z$;
	\Cref{fig:gadgetlowsimple} depicts $M_4$.	
	Let $\gamma$ and $\kappa$ be two curves crossing each other in points $s_1$ and $s_2$,
	forming an empty cell as shown in \Cref{fig:gadgetlowsimple}.
	We denote that cell with $c$.
	Furthermore, we extend $\gamma$ and $\kappa$ to infinity,
	such that each splits the plane into two areas.
	We denote the area above $\kappa$ with $c_\kappa$ and 
	the area below $\gamma$ with $c_\gamma$.
	
	To construct the drawing $M_z$ of $G_z$ we begin by drawing $z$ pseudocircles $\omega_i$, $1 \leq i \leq z$.
	Each $\omega_i$ is drawn such that it intersects $\gamma$ and $\kappa$ 
	twice in between $s_1$ and $s_2$.
	Furthermore, we require the crossings of the $\omega_i$s to appear on $\gamma$ and $\kappa$
	as we traverse them from $s_1$ to $s_2$ as follows:
	the first crossing is with $\omega_1$, the second with $\omega_2$ and so on until $\omega_z$,
	then the $z +1$-st crossing is again with $\omega_1$, the $z + 2$-nd one with $\omega_2$, and
	the $2z$-th crossing is with $\omega_z$.
Since this order is the same along $\gamma$ and $\kappa$
	we get that all crossings between the $\omega_i$s lie outside $c$.
	Furthermore, all pseudocircles pairwise intersect, once in $c_\gamma$ and once in $c_\kappa$.
	Also, by the imposed ordering there exists a point on each $\omega_i$ that lies 
	inside the cell $c$ and not in the interior of or on any $\omega_j$ with $j > i$.
	Starting from such a point in $c$ on $\omega_i$ and 
	traversing $\omega_i$ such that $\kappa$ is crossed first we require the 
	crossings with the other $\omega_i$ to appear in the same order of the indices.
	Consequently, for every $\omega_i$ with $2 < i < z$ there exists a cell lying inside $c_\gamma$ that is bounded by
	$\omega_{i-2}$, $\omega_{i-1}$, and $\omega_{i+1}$ let $c_i$ be that cell for $\omega_i$.
	For $\omega_1$, $\omega_2$, and $\omega_3$ we choose the cells bounded by 
	$\gamma$, $\kappa$, and $\omega_{i+1}$; $\kappa$, $\omega_{i-1}$, and $\omega_{i+1}$; and
	$\omega_{i-2}$, $\omega_{i-1}$, and $\kappa$, respectively.
	For each $\omega_i$ cut the pseudocircle inside $c_i$ and replace the endpoints with the vertices $u_i$ and $v_i$ in $G_z$.
	Draw the edge $u_iv_i$ by following the non-plane segment of $\omega_i$ 
	with endpoints at the positions of $u_i$ and $v_i$.
	Let $M_z$ be the such obtained drawing of edges $u_iv_i$ from $G_z$.
	
	To construct the drawing used in \Cref{thm:lowsimplicity}
	we combine the drawings from \Cref{lem:2s2p3p} with two copies $M^1_z$ and $M^2_z$ 
	of the drawing $M_z$ for $z > 0$ as follows.
	Let $u$, $v$, $w$, and $x$ be the four vertices of the four cycle and $uv$, $vw$, $wx$, and $ux$ the edges.
	Then, $uv$ crosses with $wx$ and $vw$ crosses with $ux$.
	Identify $uv$ with $\gamma$ and $wx$ with $\kappa$ in $M^1_z$ and
	$vw$ with $\gamma$ and $ux$ with $\kappa$ in $M^2_z$.
	Furthermore, we can draw the curves in each of $M^1_z$ and $M^2_z$ such that 
	no edge in $M^1_z$ intersects $vw$ or $ux$ and
	no edge  in $M^2_z$ intersects $uv$ or $wx$.
	Let $R_z$ be the resulting drawing.
	
	\begin{lemma}
		\label{lem:rz}
		For $z > 0$ the drawing $R_z$ is saturated, $2$-simple, and $2z + 2$-plane and $R_z$ has
		$4z+4$ vertices, $2z + 4$ edges, and $2(z^2+z+1)$ free cells.
	\end{lemma}
	\begin{proof}
		Let $R_z$ for $z > 0$ be a drawing obtained as described above.
		We begin by showing that $R_z$ is a saturated $2$-simple $2z + 2$-plane drawing.
		It is easy to see that $R_z$ is in fact $2$-simple and $2z + 2$-plane since
		the drawing of the four cycle is $2$-simple and $2$-plane and 
		further the edges added by $M^1_z$ and $M^2_z$ 
		each have $2z$ crossings with other edges in $M^1_z$ and $M^2_z$ respectively.
		Additionally, each edge in $M^1_z$ crosses the edges $uv$ and $wx$ twice and
		each edge in $M^2_z$ crosses the edges $vw$ and $ux$ twice,
		meaning that each of these four edges is crossed a total of $2z + 2$ times.
		Now since the drawing is $2z + 2$-plane and each edge is also crossed precisely this number of times, 
		no edge can be crossed again.
		Furthermore, no two non-adjacent vertices lie inside the same cell.
		Consequently, no edge can be added to $R_z$ and the drawing is also saturated.
		
		It remains to count the number of vertices, edges, and free cells.
		The former two are straight-froward as the four cycle contributes exactly four vertices and four edges,
		and the two copies of $M_z$ each contribute $2z$ vertices and $z$ edges. 
		Hence, $R_z$ has $n_0 = 4z + 4$ vertices and $m_0 = 2z + 4$ edges.
		To count the free cells $f_0$ we to count the free cells introduced by $M^1_z$ and $M^2_z$.
		
		Since they are symmetric we consider here the free cells added by $M^1_z$.
		The free cells incident to edge segments added to $R_z$ by $M^2_z$ can be counted in the same manner.
		We use that the edges were constructed from pseudocircles, where each two pseudocircles intersect twice,
		we find that there are a total of $z^2$ intersection points between the pseudocircles hence $z^2 + z$ cells.
		The edges $uv$ and $wx$ subdivide $2z - 1$ cells into three cells and add two more cells
		which are incident to the crossings between $uv$ and $wx$. 
		In total we obtain $z^2 + z + 2z - 1 + 2 = z^2 + 3z + 1$ many free cells.
		Finally, we have to subtract the $2z$ cells  in which we placed vertices and get that 
		edge segments added by $M^1_z$ are incident to $z^2 + z + 1$ free cells in $R_z$.
		Consequently, there are $2(z^2 + z + 1)$ free cells incident to edge segments added by
		$M^1_z$ and $M^2_z$.
		Since the two free cells of the four cycle were already counted and
		no other free cells exist this is also the total number of free cells in $R_z$.		
	\end{proof}
	
	With \Cref{lem:2s1p,lem:2s2p3p,lem:rz} we are ready to prove the main theorem of this section.
	
	\begin{restatable}{theorem}{thmtwosimple}
		\label{thm:lowsimplicity}
		For $k > 0$ there are arbitrarily large saturated drawings on $n$ vertices with
		$\frac{3(n-2)}{2}$ edges if $k =1$, $\frac{4(n-1)}{5}$ edges if $k=2$ or $k = 3$, and
		with $\frac{2(n - 1)}{k + \frac{2}{k+2}}$ edges if $k > 3$, which 
		are saturated, $2$-simple, and $k$-plane.
\end{restatable}
	\begin{proof}
		\Cref{lem:2s1p,lem:2s2p3p} proof the cases of $k=1$, $2$, and $3$.
		
		For the case of $k > 3$ we consider the drawing $R_z$ with $z > 0$.
		From \Cref{lem:rz} we know that $R_z$ is a saturated $2$-simple $2z + 2$-plane drawing.
		Consequently, if we choose $z = \frac{k-2}{2}$ the resulting drawing is in particular $k$-plane.
		Furthermore, by the same lemma, we know that $R_z$ has $n_0 = 4z+4$ vertices, 
		$m_0 = 2z + 4$ edges, and $f_0 = 2(z^2+z+1)$ free cells.
		Using \Cref{lem:stashing} we obtain a bound on the edge-vertex ratio in terms of $z$
		\begin{eqnarray*}\frac{m_0}{n_0 + f_0 - 1}= \frac{2z + 4}{4z + 4 + 2z^2 + 2z + 2 - 1}= \frac{2z + 4}{2z^2 + 6z + 5}= \frac{z + 2}{z^2 + 3z + \frac{5}{2}}.
		\end{eqnarray*}
		Substituting $\frac{k-2}{2}$ for $z$ in the above equation yields
		\begin{eqnarray*}
			\frac{z + 2}{z^2 + 3z + \frac{5}{2}}&=& \frac{\frac{k-2}{2} + 2}{(\frac{k-2}{2})^2 + 3\frac{k-2}{2} + \frac{5}{2}}\\
			&=& \frac{\frac{k-2 + 4}{2}}{\frac{k^2-4k+4}{4} + \frac{6k-12}{4} + \frac{10}{4}}\\
			&=& \frac{\frac{k+2}{2}}{\frac{k^2+2k+2}{4}}\\
			&=& \frac{2k+4}{k^2+2k+2}\\
			&=& \frac{2(k+2)}{k(k+2)+2}\end{eqnarray*}		
		as desired.		
\end{proof}
	
	{\itshape Remark.} The construction for $k > 3$ in \Cref{thm:lowsimplicity} works for every even value of $k$.
	For odd values we can begin with a saturated $2$-simple $k$-plane drawing $D$ for $k > 3$ and even.
	We add one new edge $ab$ between new vertices $a$ and $b$ which we place inside the two free cells and
	inside all edges of two $M_z$ copies.
	We draw $ab$ such that it crosses the edges $ux$ and $wx$ 
	as well as all edges in the two copies of $M_z$ in $D$.
	Then, by adding the dotted variation for $uv$ (see \Cref{fig:circlelowsimple}) we obtain a drawing $D'$ 
	in which all edges but $ab$ are crossed $k+1$ times.
	In fact, $ab$ itself has only $k$ crossings.
	This is not a problem though, as the vertices $a$ and $b$ are placed in free cells.
	Consequently, by adding one edge and two vertices as well as subtracting two free cells we obtain that $D'$
	is saturated $2$-simple $k$-plane drawing with $k > 3$ and $k$ odd and if $D'$ has $n$ vertices it has
	\begin{eqnarray*}
		\frac{2}{k + \frac{15}{k+3}-1}
	\end{eqnarray*} edges.

	\subsection{Simple drawings}
	\label{sec:simple}
	
	Auer et al.~\cite{auerSparseMaximal2013} 
	presented an elegant construction 
	for saturated simple $2$-plane drawings. 
	In this section we generalize their construction to all values $k > 0$. 
	The resulting drawings are also reminiscent of the ones in~\cite{kynclSaturatedSimple2015}.
	The construction for $k=1$ achieves the same edge-vertex ratio 
as the construction by Brandenburg et al.~\cite{brandenburgDensityMaximal2013}
	for saturated simple $1$-plane drawings. 
	
	The constructions proving \Cref{thm:kpsimple} 
	are illustrated in \Cref{fig:simple}. 
	In the figure, all the vertices on the left and on the right are identified. 
	Thus, it is more easily visualized on a vertical cylinder. 
	
	We first describe the general construction and
	proof that for all $k > 0$ it yields a saturated simple $k$-plane drawing.
	Let $k > 1$ and $t \geq 0$.
	Consider $z = t + k + 3$ vertices $u_i$ which we draw sorted from $i=1$ to $z$ 
	along a vertical line on the surface of the cylinder from top to bottom.
	We then add all edges $u_iu_{i+1}$ and $u_iu_j$ with $j = i + k + 2$ for all $1 \leq i \leq z$,
	ignoring non-existent entries.
	We also add all missing edges $u_1u_i$ for $i = 1,\ldots,k+2$ and $u_{z}u_j$ for $j = t + 1,\ldots,z$.
	
	For $k = 1$ we modify the above construction slightly.
	We consider $z = t + 6$ vertices $u_i$ and add the same edges as above.
	Additionally, we also add all possible edges $u_iu_{i+2}$ for $1 \leq i \leq z$.
	Let $S_{k,t}$ be the resulting construction.
	
	\begin{restatable}{lemma}{lemkpsaturated}
		\label{lem:kpsaturated}
		For every $k > 0$ and $t \geq 0$, $S_{k,t}$ is a saturated simple $k$-plane drawing.
	\end{restatable}
	\begin{proof}
		Clearly, the construction is simple and $k$-plane for any value of $k$ and $t$.
		Let $z$ be the number of vertices in $S_{k,t}$.
		Regardless of the values of $k$ and $t$ we find the path $u_1, u_2, \ldots, u_z$ in $S_{k,t}$.
		Furthermore, every edge in the path $u_iu_{i+1}$ for $1 < i < z - 1$ is crossed $k$ times,
		hence these edges are saturated.
		Additionally, each edge $u_iu_j$ for $j = i + k + 2$ crosses $k$ edges $u_au_b$ with 
		$i + 1 \leq a \leq i + k$ and $i + 2 \leq b \leq i + k + 1$,
		consequently these edges are saturated.
		Finally, we added all the edges incident to $u_1$ and $u_z$ to vertices with indices lower than $k+3$.
		For $k > 1$ This means, that any edge that is could potentially be added has to either cross
		an edge $u_iu_{i+1}$ or $u_iu_j$ with $j = i + k + 2$.
		Moreover, for $k = 1$ the same holds after adding the edges $u_iu_{i+2}$.
		Consequently, no edge can be added without violating $k$-planarity of the drawing and
		the $S_{k,t}$ is saturated.
	\end{proof}	
	
	\begin{figure}[bt]
		\centering
		\begin{minipage}[t]{.32\textwidth}
			\centering
			\includegraphics[page=1]{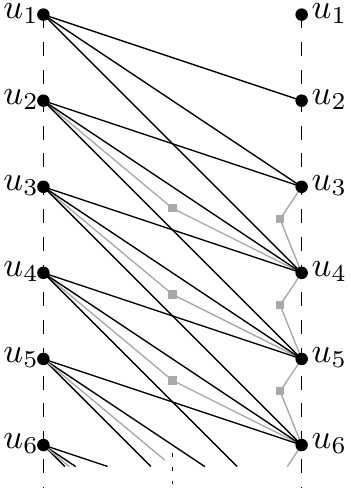}
			\subcaption{Construction for $k=1$}
			\label{fig:1psimple}
		\end{minipage}
		\hfill
		\begin{minipage}[t]{.32\textwidth}
			\centering
			\includegraphics[page=2]{simple_upperbound_cylinder}
			\subcaption{Construction for $k=2$}
			\label{fig:2psimple}
		\end{minipage}
		\hfill
		\begin{minipage}[t]{.32\textwidth}
			\centering
			\includegraphics[page=3]{simple_upperbound_cylinder}
			\subcaption{Construction for $k>2$}
			\label{fig:3psimple}
		\end{minipage}
		\caption{Construction for saturated simple $k$-plane drawings. 
			The dashed left and right sides of the drawings are identified.
}
		\label{fig:simple}
	\end{figure}
	
	We divide the proof of \Cref{thm:kpsimple} into several lemmas.
	\Cref{lem:simple1p} proves the result for saturated simple $1$-plane drawings,
	in \Cref{lem:simple2p} we show the bound for saturated simple $2$-plane drawings, 
	and in \Cref{lem:simplekp} we give the proof for saturated simple $3$-plane drawings.
	The latter two lemmas also make use of the intermediate lemma \Cref{lem:sktsize},
	which shows how many vertices and edges our construction has for $k > 1$.
	
	\begin{lemma}
		\label{lem:simple1p}
		There are arbitrarily large drawings on $n$ vertices with $\frac{7n-9}{3}$ edges,
		which are saturated, simple, and $1$-plane.
	\end{lemma}
	\begin{proof}
		Let $S_{1,t}$ be a drawing constructed as above for $t \geq 0$.
		By \Cref{lem:kpsaturated} we know that $S_{1,t}$ is saturated, simple, and $1$-plane.
		Finally, consider $k=1$ and $t \geq 0$ and let $S_{1,t}$ be the resulting construction.
		For the number of vertices we get $z = t + 6$.
		Again counting the sum of degrees $d$ we get that 
		$u_1$ and $u_t$ have degree $3$, $u_2$ and $u_{z-1}$ have degree $4$,
		$u_3$ and $u_{z-2}$ have degree $5$, and every other vertex has degree $6$.
		Hence, $d = 2(3 + 4 + 5) + 6t = 6t + 24$ and consequently there are $m_z = 3t + 12$ edges.
		
		As $S_{1,t}$ is simple and $1$-plane there are no cells with zero or only one vertex on its boundary.
		Hence, we consider cells with two vertices on their boundary.
		Stashing vertices incident to these two vertices into the cells then improves the edge-vertex ratio.
		This was also used in~\cite{brandenburgDensityMaximal2013} and~\cite{baratImprovementsDensity2018}.
		We find that every vertex $u_i$ for $i = 4,\ldots, z-3$ is on the boundary of four such cells.
		Hence, accounting for the cells that are also incident to $u_3$ and $u_{z-2}$, 
		there are $2(z - 5) + 1 = 2t + 3$ cells with only two vertices on their boundaries.
Stashing all these degree two vertices into $S_{1,t}$ we obtain the drawing $H_{1,t}$ with
		\begin{eqnarray*}
			n = t + 6 + 2t + 3 =3t + 9
		\end{eqnarray*}
		vertices and 
		\begin{eqnarray*}
			m = m_z + 2(2t + 3) = 3t + 12 + 4t + 6 = 7t + 18
		\end{eqnarray*}
		edges.
		Rearranging for $t$ we get that 
		\begin{eqnarray*}
			t = \frac{n - 9}{3}
		\end{eqnarray*}
		and hence there are
		\begin{eqnarray*}
			m = 7t + 18 = 7\frac{(n-9)}{3} + 18 = \frac{7n - 63 + 54}{3} = \frac{7n - 9}{3}
		\end{eqnarray*}
		edges.
	\end{proof}
	
	For $k > 1$ we compute the number of vertices and edges in the following lemma.
	
	\begin{lemma}
		\label{lem:sktsize}
		For $k > 1$ and $t \geq 0$ $S_{k,t}$ has $z = k + t + 3$ vertices and $m_z = 3k + 2t + 3$ edges.
	\end{lemma}
	\begin{proof}
		Let $S_{k,t}$ be as above for some $k>1$ and $t \geq 0$ 
		we get that there are $z = t + k + 3$ vertices.
		To compute the number of edges we look at the degrees of all vertices $u_i$ in $S_{k,t}$.
		For $u_1$ and $u_z$ we find that both vertices have degree $k+2$,
		for $u_2$ and $u_{z - 1}$ we get degree $3$, and
		all remaining vertices have degree four.
		In total the sum $d$ of degrees in dependence on $t$ and $k$ is
		\begin{eqnarray*}
			d = 2\cdot(k+2) + 2 \cdot 3 + 4\cdot (t + k - 1).
		\end{eqnarray*}
		Hence, there are 
		\begin{eqnarray*}
			m_z = (k+2) + 3 + 2\cdot(t + k - 1) = 3k + 2t + 3
		\end{eqnarray*}
		edges in $S_{k,t}$ if $k > 1$.
	\end{proof}
	
	\begin{lemma}
		\label{lem:simple2p}
		There are arbitrarily large drawings on $n$ vertices with $\frac{4n+7}{3}$ edges,
		which are saturated, simple, and $2$-plane.
	\end{lemma}
	\begin{proof}
		Let $S_{2,t}$ be a drawing constructed as above with $t \geq 0$.
		By \Cref{lem:kpsaturated} we know that $S_{2,t}$ is saturated, simple, and $2$-plane.
		We also know by \Cref{lem:sktsize}that it has $z =t + 5$ vertices and
		$m_z = 2t + 9$ edges.
		We see that the number of free cells per edge is clearly zero.
		Furthermore, the number of edges in $S_{2,t}$ is approximately $2n$.
		It turns out that in this situation the edge-vertex ratio can be lowered by stashing a pendant vertex
		into each cell that is bounded by edges with $k$ crossings and has only one vertex on its boundary.
		
		Let $S_{2,t}$ be a construction as above for $t \geq 0$.
		Then there are two cells with only one vertex on their boundary 
		per vertex $u_i$ with $5 \leq i \leq z - 4$,
		one such cell per vertex $u_i$ with $i \in \{2,3,4,z-3,z-2,z-1\}$, and
		$u_1$ and $u_z$ are not incident to any such cell.
		Hence, we can stash $2(z - 8) + 6 = 2(t + 5 - 8) + 6 = 2t$ pendant vertices.
		Let $H_{2,t}$ be the drawing obtained by stashing these $2t$ pendant vertices into $S_{2,t}$ and
		let $n$ be the number of vertices and $m$ the number of edges in $S_{2,t}$.
		Then, we have $n = z + 2t = 3t + 5$ and $m = m_z + 2t$.
		Rearranging for $t$ we get that 
		\begin{eqnarray*}
			t = \frac{n - 5}{3}
		\end{eqnarray*}
		and hence we obtain that there are
		\begin{eqnarray*}
			m &=& 2t + 9 + 2t = 4t + 9 
			= 4\frac{n - 5}{3} + 9
			= \frac{4n + 7}{3}
		\end{eqnarray*}
		edges.		
	\end{proof}
	
	{\itshape Remark.} Before proving the case for $k > 2$ we note that adding pendant vertices 
	would not decrease the edge-vertex ratio in the following proof.
	This can be shown similarly to the argument in \Cref{lem:stashingedge}.
	
	\begin{lemma}
		\label{lem:simplekp}
		For $k > 2$ there are arbitrarily large drawings on $n$ vertices with 
		$\frac{2n}{k-1} + \frac{3k - 2 -\frac{9}{k}}{1 - \frac{1}{k}}$ edges,
		which are saturated, simple, and $k$-plane.		
	\end{lemma}
	\begin{proof}
		Let $S_{k,t}$ be a drawing constructed as above for $k > 2$ and $t \geq 0$.
		By \Cref{lem:kpsaturated} we know that $S_{k,t}$ is saturated, simple, and $k$-plane.
		We also know by \Cref{lem:sktsize} that it has $z = k + t +3$ vertices and
		$m_z = 3k + 2t + 3$ edges.
		Consider the edges $u_iu_{i+1}$ for $i = k,\ldots,z-k-1$,
		each of these edges forms the upper boundary of $k-2$ free cells.
		Furthermore, for $k > 3$ the edges $u_iu_{i+1}$ for $i = 3,\ldots,k-1$ and $z-4,\ldots,z-k$
		bound each $1,\ldots,k-3$ free cells.
		Hence, there are
		\begin{eqnarray*}
			(k-2)(z - 2k) + 2\cdot\sum_{i=1}^{k-3}i &=& (k-2)(z-2k) + (k-3)(k-2) \\ &=& (k-2)z - k^2 - k + 6
		\end{eqnarray*}
		free cells.
By placing one isolated vertex into each free cell of $S_{k,t}$ we obtain the construction $H_{k,t}$ 
		with 
		\begin{eqnarray*}
			n &=& z + (k-2)z - k^2 - k + 6 \\
			&=& (k-1)z - k^2 - k + 6 \\
			&=& (k-1)(t + k + 3) -k^2 - k + 6 \\
			&=& kt + k^2 + 3k - t - k - 3 - k^2 - k + 6 \\
			&=& (k-1)t + k + 3
		\end{eqnarray*} 
		vertices and $m = m_z$ edges.
		Rearranging the number of vertices for $t$ gives
		\begin{eqnarray*}
			t = \frac{n-k-3}{k-1}
		\end{eqnarray*}
		and hence we obtain that there are
		\begin{eqnarray*}
			m &=& 3k + 2t + 3 \\
			&=& 3k + 2 \frac{n-k-3}{k-1} + 3 \\
			&=& \frac{3k(k-1) + 2(n-k-3) + 3(k-1)}{k-1}\\
			&=& \frac{3k^2-3k + 2n-2k-6 + 3k-3}{k-1}\\
			&=&\frac{2n}{k-1} + \frac{3k^2 - 2k -9}{k-1}\\
			&=& \frac{2n}{k-1} + \frac{3k - 2-\frac{9}{k}}{1-\frac{1}{k}}.
		\end{eqnarray*}
		edges.
	\end{proof}

	Combining \Cref{lem:simple1p,lem:simple2p,lem:simplekp} we obtain the desired theorem.

	\begin{restatable}{theorem}{thmkpsimple}
		\label{thm:kpsimple}
		There are arbitrarily large saturated simple $k$-plane drawings on $n$ vertices with: 
		$\frac{7n - 9}{3}$ edges if $k = 1$, 
		$\frac{4n + 7}{3}$ edges if $k = 2$, and 
$\frac{2n}{k-1} + \frac{3k - 2-\frac{9}{k}}{1-\frac{1}{k}}$
edges
		if $k > 2$. 
\end{restatable}

	\aclearpage

\subsection{Other drawings with few edges}
	\label{sec:lsimple}
	
	The following result completes the picture for saturated 3- and 4-simple 3- and 4-plane drawings. 
	The proof is based on the constructions depicted in \Cref{fig:34s34p}. 
	
		In this section we give the formal proofs for the constructions shown in \Cref{sec:lsimple}.
	All three are applications of \Cref{lem:stashing}.
	\begin{lemma}
		\label{lem:3s3p}
		There are arbitrarily large saturated $3$-simple $3$-plane drawings with $m = \frac{3(n - 1)}{4}$ edges.
	\end{lemma}
	\begin{proof}
		Consider the drawing shown in \Cref{fig:3s3p}.
		It is a $K_4$ drawn with an idea similar to the one used in \Cref{lem:2s2p3p}.
		Each edge is crossed three times.
		Two crossings are with adjacent edges and one with the independent edge relative to that edge.
		This creates five free cells.
		Applying \Cref{lem:stashing} with $n_0 = 4$, $f_0 = 5$, and $m_0 = 6$ gives the result.
	\end{proof}
	
	\begin{figure}[b]
		\centering
		\savebox{\smallereximagebox}{\includegraphics[page=1]{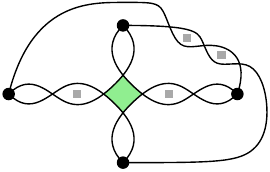}}
		\begin{minipage}[t]{.325\textwidth}
			\centering
\usebox{\smallereximagebox}
			\subcaption{$3$-simple $3$-plane: $\frac{3(n-1)}{4}$}
			\label{fig:3s3p}
		\end{minipage}
		\hfill
		\begin{minipage}[t]{.325\textwidth}
			\centering
			\raisebox{\dimexpr.5\ht\smallereximagebox-.5\height}{
				\includegraphics[page=2]{3_simple}
			}
			\subcaption{$3$-simple $4$-plane: $\frac{3(n-1)}{7}$}
			\label{fig:3s4p}
		\end{minipage}
		\hfill
		\begin{minipage}[t]{.325\textwidth}
			\centering
			\raisebox{\dimexpr.5\ht\smallereximagebox-.5\height}{
				\includegraphics[page=1]{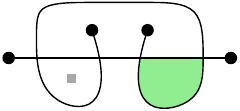}
			}
			\subcaption{$4$-simple $4$-plane: $\frac{2(n-1)}{5}$}
			\label{fig:4s4p}
		\end{minipage}	
		\caption{Constructions for \Cref{thm:smaller}.
			Gray squares represent isolated vertices that we stash and green cells are used for recursively stashing the construction.}
		\label{fig:34s34p}
	\end{figure}
	
	\begin{lemma}
		\label{lem:3s4p}		
		There are arbitrarily large saturated $3$-simple $4$-plane drawings with $m = \frac{3(n-1)}{7}$ edges.
	\end{lemma}
	\begin{proof}
		Consider the drawing shown in \Cref{fig:3s4p}.
		It is a path on three vertices $u$, $v$, $w$, and $x$ drawn such that
		$uv$ and $wx$ each cross the edge $vw$ and each other twice as shown in \Cref{fig:3s4p}.
		This creates four free cells.
		Applying \Cref{lem:stashing} with $n_0 = 4$, $f_0 = 4$, and $m_0 = 3$ gives the result.
	\end{proof}
	
	\begin{lemma}
		\label{lem:4s4p}		
		There are arbitrarily large saturated $4$-simple $4$-plane drawings with $m = \frac{2(n-1)}{5}$ edges.
	\end{lemma}
	\begin{proof}
		Consider the drawing shown in \Cref{fig:4s4p}.
		It consists of two independent edges $uv$ and $wx$ with vertices $u$, $v$, $w$, and $x$.
		The drawing is such that the two edges cross each other four times and the vertices $w$ and $x$
		are placed in one cell bounded by the two edges.
		This leaves two free cells.
		Applying \Cref{lem:stashing} with $n_0 = 4$, $f_0 = 2$, and $m_0 = 2$ gives the result.
	\end{proof}
	
	\Cref{lem:3s3p,lem:3s4p,lem:4s4p} proof \Cref{thm:smaller}.
	
	\begin{restatable}{theorem}{thmsmaller}
		\label{thm:smaller}
	There are arbitrarily large saturated $3$-simple $3$- and $4$-plane drawings on $n$ vertices with $\frac{3(n-1)}{4}$ 
	and $\frac{3(n-1)}{7}$ edges, respectively,  
and also arbitrarily large 
	 saturated $4$-simple $4$-plane 
	 drawings on $n$ vertices with $\frac{2(n-1)}{5}$ edges.
	\end{restatable}

	\section{Straight-Line Drawings}
	\label{sec:straightline}
		Finally, we consider $k$-plane drawings in which edges are drawn as straight-lines. 
We construct a family of saturated straight-line $1$-plane drawings with $n$ vertices and $\frac{11n - 12}{5}$ edges by 
		gluing $K_4$s together 
		and placing three vertices of degree two for each $K_4$, 
		as shown in \Cref{fig:1psl}.
Note that this edge-vertex ratio of $\frac{11}{5} = 2.2$ is lower than 
		the lower bound on the edge-vertex ratio for 
		saturated simple $1$-plane drawings, 
		that is $\frac{20}{9}\approx 2.22$~\cite{baratImprovementsDensity2018}.
For $k = 2$, $3$, and $4$ we take a 
		convex $4k$-gon and 
		by adding $2k$ chords as shown in \Cref{fig:3psl} for $k=3$
		we obtain a grid of free cells.
For $k > 4$ we contract two neighboring groups into one vertex each as shown in \Cref{fig:kpsl}.
		This creates two fans of $k$ edges that form a grid of free cells.
Stashing into the free cells yields a family of drawings on $n$ vertices obtaining the desired bounds.
		For $k=3$ the above construction 
gives an edge-vertex ratio of $\frac{6}{5}$, 
		but if instead of stashing isolated vertices we stash isolated edges 
		we can improve the ratio to $\frac{7}{6}$; see \Cref{fig:3psl}. 
		
For $k=1$, $2$, and $3$ and up to eight vertices 
		we tested all possible saturated straight-line $k$-plane drawings 
		using the order type database~\cite{aak-eotsp-01a}, 
		confirming for these small values that our constructions are the best possible.

		\begin{figure}[b]
			\centering
			\savebox{\straightlineimagebox}{\includegraphics[page=2]{straight_line}}
			\begin{minipage}[t]{.32\textwidth}
				\centering
				\raisebox{\dimexpr.5\ht\straightlineimagebox-.5\height}{			
					\includegraphics[page=3]{straight_line}
				}
				\subcaption{Construction for $k=1$}
				\label{fig:1psl}
			\end{minipage}
			\begin{minipage}[t]{.32\textwidth}
				\centering
\usebox{\straightlineimagebox}
				\subcaption{Construction for $k>1$}
				\label{fig:kpsl}
			\end{minipage}
			\begin{minipage}[t]{.32\textwidth}
				\centering
				\raisebox{\dimexpr.5\ht\straightlineimagebox-.5\height}{			
					\includegraphics[page=1]{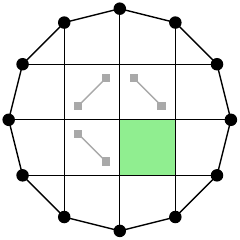}	
				}
				\subcaption{Construction for $k=3$}
				\label{fig:3psl}
			\end{minipage}		
			\caption{Constructions used in \Cref{thm:kpsl}. Gray squares represent vertices that we stash
				and green cells are used for recursively stashing the construction.}
			\label{fig:straightline}
		\end{figure}	
			
		We begin by considering the case of $k=1$ in the following lemma and
		then give the proof of \Cref{thm:kpsl}.
		\begin{lemma}
			\label{lem:1psl}
			There are arbitrarily large drawings on $n$ vertices with $\frac{11n - 12}{5}$ edges,
			which are saturated, straight-line, and $1$-plane
		\end{lemma}
		\begin{proof}
			Consider the construction shown in \Cref{fig:1psl}.
			Let $t > 1$ be an integer and place $2t + 2$ vertices on an ellipse 
			as shown in \Cref{fig:1psl}.
			We add the following edges.
			Split the vertices into an upper and lower set and let $u_i$ 
			bet the vertices in the upper and $v_i$ be the vertices in the lower set.
			We assume the vertices in either set to be labeled from zero to one as we traverse
			the upper or lower arc of the ellipse.
			Then we add the edges $u_iu_{i+1}$, $v_iv_{i+1}$,$u_iv_i$, $u_iv_{i+1}$, and $v_iu_{i+1}$ for all $0 \leq i \leq t$.
			The resulting drawing $D_t$ is clearly straight-line and $1$-plane.
			Note that for all $i$ the four vertices $u_i$,$u_{i+1}$,$v_{i+1}$, and $v_i$ form a $K_4$
			with one crossing between the edges $u_iv_{i+1}$ and $v_iu_{i+1}$.
			Then, $D_t$ is saturated as adding any edge $u_iv_{i+j}$ with $j > 1$
			has to cross the edge $u_{i+1}v_{i}$ which is already crossed by the edge $u_iv_{i+1}$.
			Moreover, we cannot add any edge $u_iu_{i+j}$ for $j>1$ as it would again cross the edge $u_{i+1}v_i$.
			The same holds for $j < 1$ using edge $u_{i-1}v_i$ and for $v_iu_{i+j}$ and $v_iv_{i+j}$.
			
			To obtain the bound we add vertices of degree two to $D_t$.		
			More precisely, we add for each $K_4$ three vertices of degree two,
			namely, if $c$ is the crossing point of the considered $K_4$ 
			in the cells bounded by $u_i$, $u_{i+1}$, and $c$,
			$u_{i+1}$,$v_{i+1}$, and $c$, and $v_i$,$v_{i+1}$, and $c$.
			Then $D_t$ has $n = 2t + 2 + 3t = 5t + 2$ vertices and $m = 5t + 1 + 6t = 11t + 2$ edges.
			With $t = \frac{n-2}{5}$ this yields $m = 11\frac{n-2}{5} + 2 = \frac{11n - 12}{5}$ for the number of edges.
		\end{proof}	
			
		\begin{restatable}{theorem}{thmkpsl}
			\label{thm:kpsl}
			For every $k>0$ there are arbitrarily large drawings on $n$ vertices with $\frac{11n - 12}{5}$ if $k = 1$,
			$\frac{3(n-1)}{2}$ if $k=2$,
			$\frac{7(n-1)}{6}$ if $k=3$,
			$n-1$ if $k=4$, and 
			$\frac{4k+2}{k^2 + 2}(n-1)$ if $k > 4$ edges,
			which are saturated, straight-line, and $k$-plane. 
		\end{restatable}
		\begin{proof}
			The case of $k=1$ is proven in \Cref{lem:1psl}.
			Next, we show how to derive a bound for $k > 1$.
			Let $s = 4k$.
			Take a cycle $C_s$ on $s$ vertices and divide the vertices into four equally large sets $A_1$,$A_2$ and $B_1$,$B_2$,
			such that their vertices are consecutive in $C_t$.
			Furthermore, we choose the sets such that as we traverse the cycle starting from the first vertex of $A_1$ 
			we encounter first all vertices of $A_1$, then of $B_1$, followed by vertices in $A_2$ and finally the ones in $B_2$.
			We label the vertices such that we encounter them in each group from index $1$ to $t$.
			Then, we draw $C_s$ into the plane such that the drawing is plane and 
			the vertices lie all on a unit circle.
			Finally, add all edges, we call them \emph{chords} below,
			$a_i^1a_i^2$ with $a^1_i \in A_1$ and $a_i^2 \in A_2$ and all edges
			$b_i^1b_i^2$ with $b^1_i \in B_1$ and $b_i^2 \in B_2$; see \Cref{fig:kpsl}.
			Let $D_0$ be the resulting drawing.
			Clearly, $D_0$ is $k$-plane.
			It is also saturated: No edge between vertices in the same group can be added as 
			it would require that there exists another vertex in the same group between them,
			but this vertex has a neighbor in the group opposite of the considered one.
			In the same way, no edge can be added between vertices in different groups without crossing a chord.
			
			The crossings of the chords in $D_0$ create a grid-like set of free cells in the center of the circle; compare also \Cref{fig:kpsl}.
			This grid has $(k-1)^2$ many cells and all of them are free as they are only bounded by edge segments of the chords.
			Using one free cell to stash the whole construction we obtain with \Cref{lem:stashing} and
			$n_0 = 4k$, $f_0 = (k-1)^2$, and $m_0 = 4k + 2k = 6k$ that
			\begin{eqnarray*}
				m = \frac{6k(n-1)}{4k + (k-1)^2 - 1} = \frac{6k(n-1)}{4k + k^2 - 2k} = \frac{6k(n-1)}{k^2 + 2k} = \frac{6(n-1)}{k + 2}.
			\end{eqnarray*}
			
			For $k = 2$ and $k = 4$ this yields the claimed bounds.
			However, for $k = 3$ we find that the resultant bound would be $\frac{6(n-1)}{3 + 2}= \frac{6(n-1)}{5}$.
			Hence, by \Cref{lem:stashingedge} we can obtain a better edge-vertex ratio by stashing new edges
			instead of isolated vertices; see \Cref{fig:3psl} for an illustration.
			Doing so, we obtain by \Cref{lem:stashingedge} and with 
			$n_0 = 4\cdot 3 = 12$, $f_0 = (3-1)^2 = 4$, and $m_0 = 6 \cdot 3  = 18$
			that for $k=3$ there are saturated straigth-line $3$-plane drawings on $n$ vertices with
			\begin{eqnarray*}
				\frac{(18 + 4 - 1)(n-1)}{12 + 2\cdot 4 - 2} = \frac{21(n-1)}{18} = \frac{7(n-1)}{6}
			\end{eqnarray*}
			edges as desired.
			
			To obtain the bound for $k > 4$ we modify the construction by contracting the vertices in $A_1$ and $B_1$ 
			into one vertex each.
			Let $a_1$ and $b_1$ be the resulting vertices and $D'_0$ the modified drawing.
			The vertices $a_1$ and $b_1$ are incident to $k$ edges each that were previously inserted between
			vertices in $A_1$ and $A_2$ and $B_1$ and $B_2$.
			Since we did not change the groups $A_2$ and $B_2$ these edges still cross and 
			hence there are also $(k-1)^2$ free cells in $D'_0$.
			With the number of vertices as $n'_0 = 2k + 2$ and 
			the number of edges in $D'_0$ as $m'_0 = 2k + 2 + 2k = 4k + 2$ we obtain
			that there are arbitrary large saturated straight-line $k$-plane drawings with 
			\begin{eqnarray*}
				\frac{(4k+2)(n-1)}{2k + 2 + (k-1)^2 - 1} = \frac{(4k+2)(n-1)}{2k + 2+ k^2 - 2k} 
				= \frac{4k+2}{k^2 + 2}(n-1)
			\end{eqnarray*}
			edges as desired.
		\end{proof}	

\section{Saturated $k$-Plane Drawings of Matchings}
	\label{sec:matchings}
	Throughout this section we consider multi-graphs, i.e.,
	two vertices can be connected with more than one edge and
	we disallow self-intersecting edges.
	To clearly distinguish this setting from the previous sections,
	we call saturated drawings in which inserting parallel edges is allowed
	\emph{multi-saturated}.

	Chaplick et al.~\cite{chaplickEdgeMinimumSaturated20} presented multi-saturated simple $k$-plane drawings
	of arbitrarily large matchings for $k \geq 7$.
	They also rule out the existance of such drawings for any $k\leq 3$.
	In this section we resolve the remaining open cases of $k=4$, $5$, and $6$.
	We prove that for $k=4$ there are no multi-saturated $k$-plane drawings of arbitrarily large matchings,
	regardless of the simplicity, i.e., only self-intersecting edges are not allowed.
	For the case of $k=5$ we show that multi-saturated $k+1$-simple $k$-plane drawings of arbitrarily large matchings exist,
	but no multi-saturated simple $k$-plane drawings.
	Finally, we present a construction for multi-saturated simple $6$-plane drawings.

	\begin{lemma}
		\label{lem:matching4}
		There are no multi-saturated $4$-plane drawings of arbitrarily large matchings without self-intersecting edges.
	\end{lemma}
	\begin{proof}
Let $G = (V,E)$ be a matching on $n$ vertices and $m$ edges,
		$D(G)$ a multi-saturated $4$-plane drawing of $G$ without self-intersecting edges, 
		$f$ the number of cells in $D(G)$ and
		$x_i$ the number of edges that have $i \in \mathbb{N}_0$ crossings in $D(G)$.
		Consider the planarization $\mathcal D = (P,C)$ of $D(G)$.
To simplify the following argumentation we remove each vertex in $V$ and its unique incident edge from $\mathcal D$.
		Note that this does not change the number of cells $f$.
		We know that $|P| = (4x_4 + 3x_3 + 2x_2 + x_1)/2$ and $|C| = 3x_4 + 2x_3 + x_2$.
		With Euler's formula it now follows that 
		\begin{eqnarray*}
			f &=& 3x_4 + 2x_3 + x_2 - (4x_4 + 3x_3 + 2x_2 + x_1)/2 + 2\\
			  &=& x_4 + \frac{1}{2}(x_3 - x_1) + 2.
		\end{eqnarray*}
		Moreover, we know that $f \geq 2 \sum_{i = 0}^4 x_i$ since every vertex in $V$ has to be part of a distinct cell in $D(G)$.
		Consquently, we obtain that
		\begin{eqnarray}
			\label{eq:match4euler}
			2(x_4 + x_3 + x_2 + x_1 + x_0) &\leq& x_4 + \frac{1}{2}(x_3 - x_1) + 2\nonumber\\
			4x_4 + 4x_3 + 4x_2 + 4x_1 + 4x_0 &\leq& 2x_4 + x_3 - x_1 + 4\\
			2x_4 + 3x_3 + 4x_2 + 5x_1 + 4x_0 &\leq& 4.\nonumber
		\end{eqnarray}
		The only solution to this equation that do not result in an empty graph or a graph containing only one edge is
		\begin{eqnarray*}
x_4 &=& 2 \text{ and } x_i = 0 \text{ for } i = 0,1,2,3.
		\end{eqnarray*}
This solution leads to an equality between the left-hand and right-hand side in Equation~\ref{eq:match4euler},
		which implies that the only possible connected multi-saturated $4$-plane drawings of a matching have two edges and
		as many cells as there are endpoints of edges, see Figure~\ref{fig:matching4} for an example.
		Stashing into such a drawing is not possible and hence the lemma follows.
	\end{proof}

	\begin{figure}[t]
		\begin{minipage}[t]{.48\textwidth}
			\centering
			\includegraphics[page=1]{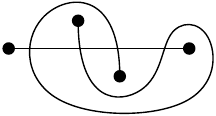}
			\caption{A multi-saturated $4$-plane drawing of a matching with two edges.
			Since no cell is free the drawing cannot be stashed.}
			\label{fig:matching4}
		\end{minipage}
		\hfill
		\begin{minipage}[t]{.48\textwidth}
			\centering		
			\includegraphics[page=1]{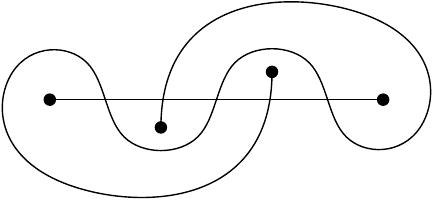}
			\caption{A multi-saturated $5$-plane drawing of a matching with two edges.
			Stashing is possible in the outer cell.}
			\label{fig:matching5twoedges}
		\end{minipage}
	\end{figure}

	\begin{figure}[t]
		\begin{minipage}[t]{.48\textwidth}
			\centering
			\includegraphics[page=2]{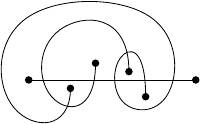}
			\caption{A multi-saturated $5$-plane drawing of a matching with three edges.
			Since no cell is free the drawing cannot be stashed.}
			\label{fig:matching5threeedges}
		\end{minipage}
		\hfill
		\begin{minipage}[t]{.48\textwidth}
			\centering
			\includegraphics[page=3]{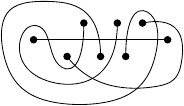}
			\caption{A multi-saturated $5$-plane drawing of a matching with four edges.
			Since no cell is free the drawing cannot be stashed.}
			\label{fig:matching5fouredges}
		\end{minipage}
	\end{figure}

	\begin{lemma}
		\label{lem:matching5}
		There are no multi-saturated $\ell$-simple $5$-plane drawings of arbitrarily large matchings for $\ell < 6$.
	\end{lemma}
	\begin{proof}
		The proof follows the same idea as the one for Lemma~\ref{lem:matching4}.
		Let $G = (V,E)$ be a matching, $D(G)$ a multi-saturated $5$-plane drawing of it,
		$f$ the number of cells in $D(G)$, and
		$x_i$ the number of edges that have $i \in \mathbb{N}_0$ crossings in $D(G)$.
		Again, we consider the planarization $\mathcal D = (P,C)$ of $D(G)$ with the vertices from $V$ and their incident edge removed.
		We know that $|P| = (5x_5 + 4x_4 + 3x_3 + 2x_2 + x_1)/2$ and $|C| = 4x_5 + 3x_4 + 2x_3 + x_2$.
		Using Euler's formula in the same manner as above we obtain that
\begin{eqnarray*}
x_5 + 2x_4 + 3x_3 + 4x_2 + 5x_1 + 4x_0 &\leq& 4.
		\end{eqnarray*}
		The only solutions to this equation not resulting in an empty graph, a graph with only one edge, or
		a non-integer number of crossings are
		\setcounter{equation}{0}
		\begin{eqnarray}
x_i &=& 1 \text{ for } i = 3,5 \text{ and } x_j = 0 \text{ for } j = 0,1,2,4\\
x_4 &=& 1 \text{ and } x_5 = 2 \text{ and } x_i = 0 \text{ for } i = 0,1,2,3\\
			x_4 &=& 2 \text{ and } x_i = 0 \text{ for } i = 0,1,2,3,5\\
			x_5 &=& 2 \text{ and } x_i = 0 \text{ for } i = 0,1,2,3,4\\
x_5 &=& 4 \text{ and } x_i = 0 \text{ for } i = 0,1,2,3,4.
		\end{eqnarray}
Solution~(1) implies that at least one edge has to self-intersect, which is not allowed in our present setting.
Any drawing realizing Solution~(3), i.e., two edges with four crossings each, is not multi-saturated as
		either edge can be crossed another time while every cell of the drawing must contain a vertex.
		Solutions~(2) and~(5) imply
connected drawings with as many cells as there are endpoints of edges.
		Hence, stashing in them is not possible.
Drawings realizing Solutions~(2) and~(5) can be seen in Figure~\ref{fig:matching5threeedges} and~\ref{fig:matching5fouredges} respectively.
		Finally, Solution~(4) leads to a drawing with five crossings and two edges per connected component.
		In Figure~\ref{fig:matching5twoedges} we give an example of a multi-saturated $6$-simple $5$-plane drawing realizing this solution.
		As two edges with five crossings each can never result in a $5$-simple drawing the lemma follows.
	\end{proof}

	\begin{figure}[t]
		\centering
		\begin{minipage}[t]{.48\textwidth}
			\includegraphics[page=2]{figures/matching_6.pdf}
			\caption{A multi-saturated simple $6$-plane drawing of a matching with only seven edges.}
			\label{fig:matching6}
		\end{minipage}
		\hfill
		\begin{minipage}[t]{.48\textwidth}
			\includegraphics[page=1]{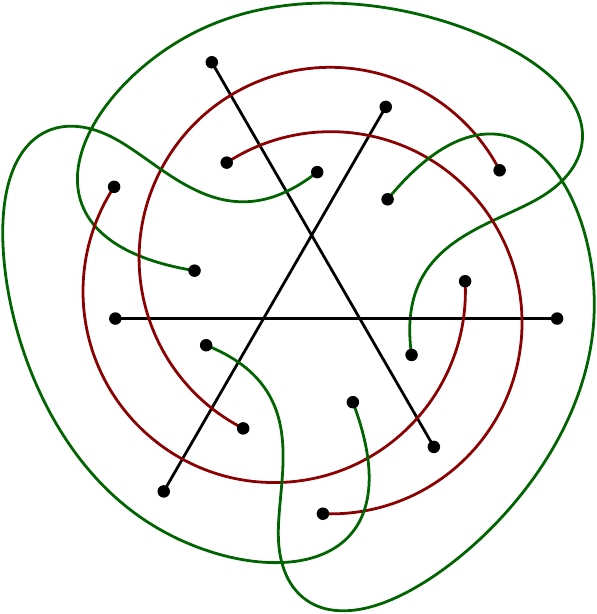}
			\caption{A multi-saturated simple $6$-plane drawing of a matching with nine edges. 
			The edge colors indicate edges that behave symmetrically.}
			\label{fig:matching6nine}
		\end{minipage}

	\end{figure}

	\begin{theorem}
		\label{thm:matchings}
		There are no multi-saturated simple $k$-plane drawings of matchings for $k \in \{4,5\}$ on $n$ vertices,
		while such drawings exist for $k = 6$. 
		Moreover, for $k=6$ we can construct such drawings of arbitrarily large matchings.
	\end{theorem}
	\begin{proof}
		The claims for $k=4$ and~$5$ follow directly from Lemmas~\ref{lem:matching4} and~\ref{lem:matching5}.

		Figure~\ref{fig:matching6} shows an example of a matching of seven edges drawn with exactly six crossings per edge.
		Clearly, the drawing is also simple as no two edges cross each other more than once.
		Moreover, every vertex lies in a distinct cell and no two lie in the same cell.
		Let $u$ be an arbitrary vertex of the drawing, 
		then we cannot add an edge between $u$ and any other vertex in the drawing 
		without crossing the cell boundaries of the cell containing $u$.
		Yet, each edge on that boundary is already crossed six times.
	\end{proof}

	The matching in Figure~\ref{fig:matching6} is also as sparse as possible, 
	matching the lower bound by Chaplick et al.~\cite{chaplickEdgeMinimumSaturated20}.
	In Figure~\ref{fig:matching6nine} we show a fully symmetric saturated simple $6$-plane drawing of a matching with nine edges.
	Exploiting its symmetry one can generate saturated simple $6$-plane drawings of arbitrarily large matchings without
	relying on stashing (i.e. the planarization is a connected graph).

	\section{Conclusion}
	\label{sec:conclusion}
	With this paper we initiated the study of saturated drawings in the context of $k$-planarity.
	We presented constructions depending on the simplicity of the drawing and 
	translated results from the study of maximal $1$- and $2$-planar graphs.
	The, in our opinion, the most interesting open problems are 
	tightening the bounds for saturated $k+1$-simple and simple $k$-plane drawings.
	In particular, achieving a tight bound for saturated simple $1$-plane drawings.
	Our results show that as we reduce the simplicity the number of edges we require to construct
	saturated $k$-plane drawings seems to increase.

\bibliography{saturated_k_planar.bib}
	
	\clearpage
	\appendix

	\section{Deriving Euler's formula for disconnected plane graphs}
	\label{apx:euler}
	For completeness, we show in this section how to derive the version of Euler's formula 
	used in the proof of the lower bound of \Cref{thm:selfcrossings}.
	Let $\mathcal G = (P,C)$ be an embedded planar graph with $\gamma' > 0$ many connected components and
	$f$ the number of faces in $\mathcal G$.
	Usually Euler's formula is stated as 
	\begin{eqnarray*}
		|P| - |C| + f = 2
	\end{eqnarray*}
	and it is assumed that the given graph is simple and connected.
	In our case $\mathcal G$ might be neither simple nor connected.
	More precisely, we allow a vertex in $\mathcal G$ to have at most one self-loop or 
	be involved in one pair of multi-edges.
We count these self-loops and multiple edges by adding one to the number of faces for each occurrence of either.
	As a result, the number of edges introduced per self-loop and multiple edge just cancels with the additionally added faces.
	
	To handle the connected components of $\mathcal G$ note that
	every such component is itself a connected embedded planar graph.
	Let $P_i$, $C_i$, and $f_i$ be the vertices, edges, and the number of faces 
	for the $i$-th connected component of $\mathcal G$, then $|P_i| - |C_i| + f_i = 2$ holds for each $1 \leq i \leq \gamma'$.
	Summing over all $i$ we get
	\begin{eqnarray*}
		\sum_{i = 1}^{\gamma'}|P_i| - \sum_{i = 1}^{\gamma'}|C_i| + \sum_{i = 1}^{\gamma'}f_i = 2\gamma'.
	\end{eqnarray*}
	Note that since we count the outer face for each component we have
	\begin{eqnarray*}
		\sum_{i = 1}^{\gamma'}f_i = f + (\gamma' - 1).
	\end{eqnarray*}
	Consequently we obtain
	\begin{eqnarray*}
		\sum_{i = 1}^{\gamma'}|P_i| - \sum_{i = 1}^{\gamma'}|C_i| + \sum_{i = 1}^{\gamma'}f_i &=& 2\gamma'\\
		|P| - |C| + f + (\gamma' - 1) &=& 2\gamma' \\
		|P| - |C| + f &=& \gamma' + 1.
	\end{eqnarray*}

\end{document}